\documentclass[reqno]{amsart}
\usepackage{amssymb}
\usepackage{amsthm}
\usepackage{color}
\usepackage{eucal}
\usepackage{eufrak}
\usepackage{bbm,dsfont}

\newtheorem{definition}{Definition}
\newtheorem{lemma}{Lemma}
\newtheorem{proposition}{Proposition}
\newtheorem{theorem}{Theorem}
\newtheorem{cor}{Corollary}

\newcommand{\nat}{\mathbb N} 
\newcommand{\half}{\tfrac{1}{2}} 

\newcommand{\hi}{\mathcal{H}} 
\newcommand{\hik}{\mathcal{K}} 
\newcommand{\lin}[1]{\mathcal{L}(#1)} 
\newcommand{\lh}{\lin{\hi}} 
\newcommand{\trh}{\mathcal{T(H)}} 
\newcommand{\trhh}{\mathcal{T(\hi\otimes\hi)}} 
\newcommand{\tr}[1]{\mathrm{tr}\left[#1\right]} 
\newcommand{\id}{\mathbbm{1}} 

\newcommand{\Ec}{\mathcal{E}}

\newcommand{\la}{\langle}
\newcommand{\ra}{\rangle}

\newcommand{\mc}[1]{\mathcal{#1}}

\newcommand{\mb}[1]{\mathbb{#1}}
\newcommand{\mr}[1]{\mathrm{#1}}
\newcommand{\ms}[1]{\mathsf{#1}}

\newcommand{\sis}[2]{\la #1|#2\ra}
\newcommand{\Sis}[2]{\big\la #1\big| #2\big\ra}

\newcommand{\ktb}[2]{| #1\ra\la #2|}
\newcommand{\hil}{\mc H}
\newcommand{\f}{\varphi}

\newcommand{\Om}{\Omega}

\title[Extremal marginals of a completely positive map]{When do pieces determine the whole? \\
Extremal marginals of a completely positive map}

\author{Erkka Haapasalo}
\author{Teiko Heinosaari}
\author{Juha-Pekka Pellonp\"{a}\"{a}}

\begin{document}

\begin{abstract}
We will consider completely positive maps defined on tensor products of von Neumann algebras and taking values in the algebra of bounded operators on a Hilbert space and particularly certain
convex subsets of the set of such maps. We show that when one
of the marginal maps of such a map is an extremal point, then the marginals uniquely determine the map.
We will further prove that when both of the marginals are extremal, then the whole map is extremal. 
We show that this general result is the common source of several well-known results dealing with, e.g., jointly measurable observables. 
We also obtain new insight especially in the realm of quantum instruments and their marginal observables and channels.
\end{abstract}

\maketitle

\section{Introduction}

A typical feature of quantum theory is that we have many layers of descriptions, differing by their amount of details.
Usually the choice of the descriptions depends on our needs or on available information.
We can use highly complete or total description of some device or setting, but also less detailed description of the same thing if that is more suitable for our purposes or we simply lack information.
A prototypical example is joint states versus reduced states. 
A joint state describes a state of a composite system, whereas a reduced state describes a state of one of the subsystems. 
A generic joint state contains more information than the related reduced states, and even if we know the reduced states of all subsystems, we may not know the joint state.

In a typical case lower layer description (e.g.\  reduced states) does not uniquely determine higher layer description (joint state), since variation is greater in the latter and even big differences can lead to the same coarser descriptions. 
However, in some special cases pieces can be combined into a whole in a unique way.
In this work we identify one such condition to be extremality of a piece and we unify several results that fall under this theme.
By an extremal object we mean an extremal point in the convex set of all similar objects. (This total set can be e.g.\  the set of states, observables, instruments or channels.)

We will consider objects consisting of two pieces.
We show that, in a well specified setting, 
\begin{quote}
\emph{when one of the pieces is an extremal object, then pieces uniquely determine the whole}.
\end{quote}
We will further prove that
\begin{quote}
\emph{when both of the pieces are extremal objects, then the whole is an extremal object}.
\end{quote}

We recall the following well-known results that exemplify the previously sketched ideas.

\begin{itemize}

\item[(a)] \emph{Joint state with a pure marginal state}: Suppose that $\varrho$ is a state of a composite system $\hil_1\otimes\hil_2$. If one of the reduced states $\mr{tr}_{\hil_2}[\varrho]\equiv\varrho_1$ or $\mr{tr}_{\hil_1}[\varrho]\equiv\varrho_2$ is pure, then $\varrho=\varrho_1\otimes\varrho_2$.
If both $\varrho_1$ and $\varrho_2$ are pure, then also $\varrho$ is pure.

\item[(b)] \emph{Joint observable with a sharp marginal observable}: Suppose that $\ms{M}$ and $\ms{N}$ are jointly measurable observables (POVMs). 
If $\ms{M}$ or $\ms{N}$ is sharp (i.e.\  projection valued measure), then their joint observable $\ms{J}$ is unique and it is determined by the condition $\ms{J}(X\times Y)=\ms{M}(X)\ms{N}(Y)$ for all outcome sets $X,Y$.
(See e.g.\  \cite{HeReSt08} for a proof of this fact.)

\item[(c)] \emph{Instruments related to a sharp observable}: Suppose that an observable $\ms{M}$ is sharp and $\Gamma$ is an instrument such that $\Gamma(X,\id)=\ms{M}(X)$ for all outcome sets $X$.
Then $\Gamma(X,A)=\ms M(X)\mc E(A)$, where $\mc E=\Gamma(\Om,\cdot)$ is the total state transformation and $\Om$ is the total set \cite{Ozawa84}. Hence the instrument $\Gamma$ is completely determined by its total state transformation $\mc E$. 

\item[(d)] \emph{Variant of `No Cloning Theorem'}: Suppose $\mc F:\trh\to\trhh$ is a quantum channel such that $\mr{tr}_{1}[\mc F(\varrho)]=\varrho$ for every state $\varrho$. 
Then $\mr{tr}_{2}[\mc F(\varrho)]\equiv\sigma$ for some fixed state $\sigma$, hence the attempted copy $\mr{tr}_{2}[\mc F(\varrho)]$ contains no information on the input state $\varrho$.

\end{itemize}

\noindent Our main result contains all these statements as corollaries and identifies the common source behind the uniqueness claims  as being extremality of a marginal map.
We will also demonstrate some new applications of this result.
Our main theorem implies the following:

\begin{itemize}

\item[(e)] Suppose that $\ms{M}$ and $\ms{N}$ are jointly measurable observables (POVMs). 
If $\ms{M}$ or $\ms{N}$ is extremal, then their joint observable is unique.
If both $\ms{M}$ and $\ms{N}$ are extremal, then their unique joint observable is extremal.

\item[(f)] Suppose that an observable $\ms{M}$ and a channel $\Ec$ are parts of a single instrument $\Gamma$, i.e., $\Gamma(X,\id)=\ms{M}(X)$ for all outcome sets $X\subseteq\Omega$ and
$\Gamma(\Omega,\cdot)=\Ec$ for the total set $\Omega$.
If $\ms{M}$ or $\Ec$ is extremal, then the instrument $\Gamma$ is unique.
If $\ms{M}$ and $\Ec$ are both extremal, then $\Gamma$ is extremal.

\end{itemize}

Our investigation is organized as follows.
In Section \ref{sec:prel} we fix the notation and recall some standard results.
In Section \ref{sec:ext} we derive a general criterion for extremality.
Our main result, sketched here in the introduction, is proved in Section \ref{sec:joint}. In Section \ref{sec:comp} we concentrate on quantum channels on compound systems and prove that causal channels are local if one of the channels reduced to one of the subsystems is extremal.
Finally, in Section \ref{sec:end} we summarize our conclusions and make some final remarks.

\section{Notation and preliminaries}\label{sec:prel}

We denote by $\mc A$ a von Neumann algebra with the unit $\id_{\mc A}$ and $\hil$ is a complex Hilbert space with the inner product $\sis{\cdot}{\cdot}:\hil\times\hil\to\mb C$. 
We denote the set of bounded linear operators on $\hil$ by $\lh$ and the unit of $\lh$ by $\id_{\hil}$. 
The subscripts of the unities may be omitted if there is no ambiguity of the algebra or Hilbert
space. 
Although we concentrate on von Neumann algebras, it should be noted that Theorem \ref{dilat} and Proposition \ref{RNd} of this section also hold in the more general
setting of $C^*$-algebras. We follow the convention $\nat=\{1,\,2,\,3,\ldots\}$.

We recall that a linear map $\Phi:\mc A\to\mc L(\hil)$ is {\it completely positive} (CP) if for any $n\in\mb N$, $a_1,\ldots,\,a_n\in\mc A$ and $\f_1,\ldots,\,\f_n\in\hil$ we have
\begin{equation}
\sum_{j,k=1}^n\sis{\f_j}{\Phi(a_j^*a_k)\f_k}\geq0.
\end{equation}
We denote the set of all CP maps $\Phi:\mc A\to\mc L(\hil)$ by ${\bf CP}(\mc A;\hil)$.
We recall the following fundamental result for CP maps \cite{Stinespring55}.

\begin{theorem}\label{dilat}
For any $\Phi\in{\bf CP}(\mc A;\hil)$ there is a triple $(\mc M,\pi,J)$ where $\mc M$ is a Hilbert space, $\pi:\mc A\to\mc L(\mc M)$ is a unital *-representation and $J:\hil\to\mc M$ is a linear map
such that
\begin{itemize}
\item[{\rm (i)}] $\Phi(a)=J^*\pi(a)J$ for all $a\in\mc A$ and
\item[{\rm (ii)}] the linear span of the set $\{\pi(a)J\f\,|\,a\in\mc A,\ \f\in\hil\}$ is dense in $\mc M$.
\end{itemize}
If $(\mc M',\pi',J')$ is another such triple, there is a unitary operator $U:\mc M\to\mc M'$ such that $U\pi(a)=\pi'(a)U$ for all $a\in\mc A$ and $UJ=J'$.
\end{theorem}

\begin{definition}
{\rm
Suppose that $\Phi\in{\bf CP}(\mc A;\hil)$ is associated with a triple $(\mc M,\pi,J)$ satisfying the conditions (i)--(ii) of  Theorem \ref{dilat}. 
The triple $(\mc M,\pi,J)$ is called a {\it minimal Stinespring dilation for $\Phi$}. 
}
\end{definition}

If the Hilbert space $\hil$ and the algebra $\mc A$ are separable, the dilation space $\mc M$ of a minimal Stinespring dilation of a CP map $\Phi\in{\bf CP}(\mc A;\hil)$ can be chosen to be
separable \cite[Section 9.2]{QTOS76}. This is the case in the typical physical applications; see the end of this section.

Suppose that $\Phi,\,\Psi\in{\bf CP}(\mc A;\hil)$. We denote $\Psi\leq\Phi$ or $\Phi\geq\Psi$ when $\Phi-\Psi\in{\bf CP}(\mc A;\hil)$. This relation is a partial order in ${\bf CP}(\mc A;\hil)$.
The following result has been proved in \cite{Raginsky03}.
We provide a proof for the reader's convenience.

\begin{proposition}\label{RNd}
Suppose that $\Phi,\,\Psi\in{\bf CP}(\mc A;\hil)$ and $\Psi\leq\Phi$. Also assume that $(\mc M,\pi,J)$ is a minimal Stinespring dilation of $\Phi$. There is a unique operator $E\in\mc L(\mc M)$ such
that $[E,\pi(a)]=0$ for all $a\in\mc A$ and
\begin{equation}\label{deriv}
\Psi(a)=J^*\pi(a)EJ  \qquad \forall a\in\mc A \, .
\end{equation}
\end{proposition}

\begin{proof}
Suppose that $(\mc M',\pi',J')$ is a minimal Stinespring dilation of $\Psi$. Pick $n\in\mb N$, $a_1,\ldots,\,a_n\in\mc A$ and $\f_1,\ldots,\,\f_n\in\hil$ and set
\begin{equation}\label{eq:eta}
\eta=\sum_{j=1}^n\pi(a_j)J\f_j\in\mc M,\qquad\zeta=\sum_{j=1}^n\pi'(a_j)J'\f_j\in\mc M'.
\end{equation}
Using $\Psi\leq\Phi$ we may evaluate
\begin{eqnarray*}
\|\zeta\|^2&=&\sum_{j,k=1}^n\sis{J'\f_j}{\pi'(a_j^*a_k)J'\f_k}=\sum_{j,k=1}^n\sis{\f_j}{\Psi(a_j^*a_k)\f_k}\leq\sum_{j,k=1}^n\sis{\f_j}{\Phi(a_j^*a_k)\f_k}\\
&=&\sum_{j,k=1}^n\sis{J\f_j}{\pi(a_j^*a_k)J\f_k}=\|\eta\|^2.
\end{eqnarray*}
Using property (ii) of $(\mc M,\pi,J)$ in Theorem \ref{dilat} we define an operator $D:\mc M\to\mc M'$ such that $D\pi(a)J\f=\pi'(a)J'\f$ for all $a\in\mc A$ and $\f\in\hil$; by the above calculation $D$ is well-defined and bounded.
It follows that 
\begin{equation}
\Psi(a)=(J')^*\pi'(a)J'=J^*D^*D\pi(a)J
\end{equation}
for all $a\in\mc A$. 
Moreover for all $a,\,b\in\mc A$ and $\f\in\hil$ one has
\begin{equation*}
D\pi(a)\pi(b)J\f=D\pi(ab)J\f=\pi'(ab)J'\f=\pi'(a)\pi'(b)J'\f=\pi'(a)D\pi(b)J\f \, .
\end{equation*}
The minimality of $(\mc M,\pi,J)$ yields thus $D\pi(a)=\pi'(a)D$ for all $a\in\mc A$. 
Hence also $\pi(a)D^*=D^*\pi'(a)$ for all $a\in\mc A$. 
Define $E:=D^*D$ and, using the results above, we see that $E$ is as in the claim.

Suppose now that $E_1,\,E_2\in\mc L(\mc M)$ are as in the claim. 
Suppose that $\eta\in\mc M$ is the vector defined in \eqref{eq:eta}. 
We obtain
\begin{eqnarray*}
\sis{\eta}{(E_1-E_2)\eta}&=&\sum_{j,k=1}^n\sis{J\f_j}{\pi(a_j^*)(E_1-E_2)\pi(a_k)J\f_k}\\
&=&\sum_{j,k=1}^n\sis{J\f_j}{\pi(a_j^*a_k)(E_1-E_2)J\f_k}\\
&=&\sum_{j,k=1}^n\big(\sis{\f_j}{\Psi(a_j^*a_k)\f_k}-\sis{\f_j}{\Psi(a_j^*a_k)\f_k}\big)=0 \, .
\end{eqnarray*}
Again, from the minimality of $(\mc M,\pi,J)$ follows that $E_1=E_2$.
\end{proof}

Any $\Phi\in{\bf CP}(\mc A;\hil)$ is guaranteed to be norm continuous. If $\Phi$ is additionally ultraweakly continuous, i.e., is continuous with respect to the ultraweak topologies of
$\mc A$ and $\mc L(\hil)$, we say that $\Phi$ is {\it normal}. We denote the set of normal CP maps $\Phi\in{\bf CP}(\mc A;\hil)$ by ${\bf NCP}(\mc A;\hil)$. Whenever $\Phi\in{\bf NCP}(\mc A;
\hil)$, there is a unique predual map $\Phi_*:\mc T(\hil)\to\mc A_*$, where $\mc T(\hil)$ is the set of trace class operators on $\hil$ (the predual of $\mc L(\hil)$) and $\mc A_*$ is the predual of
$\mc A$, which is continuous with respect to the trace norm topology of $\mc T(\hil)$ and the norm topology of $\mc A_*$. The maps $\Phi$ and $\Phi_*$ are related by
\begin{equation}
\tr{T\Phi(a)}=\la\Phi_*(T),a\ra,\qquad T\in\mc T(\hil) \, , \quad a\in\mc A \, , 
\end{equation}
where $\la\cdot,\cdot\ra:\mc A_*\times\mc A\to\mb C$ is the canonical bilinear form (dual pairing).
A typical physically relevant CP map $\Phi$ is normal and hence has the predual $\Phi_*$.
These two maps $\Phi$ and $\Phi_*$ are considered as different representations of the same object, and $\Phi$ is said to be in the {\it Heisenberg picture} while $\Phi_*$ in the {\it Schr\"odinger picture}.

Since $\mc A$ is a von Neumann algebra, Theorem \ref{dilat} can be augmented with the following additional result \cite[Section 9.2]{QTOS76}: 
{\it For any $\Phi\in{\bf NCP}(\mc A;\hil)$, there is a minimal Stinespring dilation $(\mc M,\pi,J)$ where the *-representation $\pi$ is normal.}

\begin{definition}
{\rm
For each positive operator $P\in\mc L(\hil)$, we denote by ${\bf CP}_P(\mc A;\hil)$ the set of maps $\Phi\in{\bf CP}(\mc A;\hil)$ satisfying $\Phi(\id_{\mc A})=P$. 
We further denote ${\bf NCP}_P(\mc A;\hil)={\bf CP}_P(\mc A;\hil) \cap {\bf NCP}(\mc A;\hil)$.
}
\end{definition}

The set ${\bf NCP}_P(\mc A;\hil)$ has different functions in quantum mechanics depending on the choices for the algebra $\mc A$ and the operator $P$. 
In the following we recall the most relevant ones. The Hilbert spaces (input and output spaces; typically denoted by $\hil$ or $\mc K$) appearing in the definitions of the sets of quantum
mechanical CP maps are always assumed to be separable. 
In most of the cases the dilation spaces of these CP maps are also automatically separable.

\begin{itemize}

\item \emph{States}:
We say that an operator $\varrho\in\mc L(\hil)$ is a {\it state} if it is a positive trace class operator and $\tr{\varrho}=1$. 
We denote the set of states on $\hil$ by $\mc S(\hil)$. 
From the duality $\lh_\ast=\trh$ follows that $\mc S(\hil)={\bf NCP}_1(\mc L(\hil);\mb C)$.

\item \emph{Channels}: 
Suppose that $\mc K$ is a separable Hilbert space. 
An element in the set ${\bf NCP}_{\id}(\mc L(\mc K);\hil) \equiv \mc C(\mc K,\hil)$ is called a \emph{channel}.
A channel corresponds to a procedure where a system described by the Hilbert space $\hil$ transforms
into a system described by the space $\mc K$. 
The dilation space appearing in a minimal Stinespring dilation of a channel is separable since $\hi$ and $\hik$ are separable.

\item \emph{Observables}:
Suppose that $(\Om,\Sigma)$ is a measurable space, i.e., $\Om\neq\emptyset$ and $\Sigma$ is a $\sigma$-algebra of subsets of $\Om$. Also assume that $\nu:\Sigma\to[0,\infty]$ is a
$\sigma$-finite measure. 
An element in the set ${\bf NCP}_{\id}(L^\infty(\nu);\hil)\equiv\mc O_\nu(\hil)$ is called an \emph{observable}.
Observables that are *-homomorphisms are called sharp observables \cite{OQP97}. 
Usually observables are identified with normalized positive operator valued measures (POVMs) in the following way: for any $M\in\mc O_\nu(\hil)$ and $X\in\Sigma$, we denote $\ms M(X)=M(\chi_X)$, where $\chi_X$ is the characteristic function of $X$. Now $\ms M$ is a POVM and $M(f)=\int f\,d\ms M$ for all
$f\in L^\infty(\nu)$. 
Conversely, every POVM is of the above form for some measure $\nu$ and $M\in \mc O_\nu(\hil)$ \cite{Part1}. From now on, we view observables as POVMs and we denote POVMs defined on $\Sigma$ by $\mc O_\Sigma(\hil)$. In this view, sharp observables are projection-valued measures. The minimal
dilation of a POVM is often called a minimal {\it Na\u{\i}mark dilation}. The separability or non-separability of the dilation space of a minimal Stinespring dilation of a POVM depends on the properties of the
value space $(\Om,\Sigma)$. 
For instance, if the $\sigma$-algebra $\Sigma$ is numerably generated, then the algebra $L^\infty(\nu)$ with any $\sigma$-finite measure $\nu:\Sigma\to[0,\infty]$ is separable and hence the
dilation space in a minimal Stinespring dilation is guaranteed to be separable.

\item \emph{Instruments}:
Suppose that $(\Om,\Sigma)$ is a measurable space and $\nu:\Sigma\to[0,\infty]$ is a
$\sigma$-finite measure. 
An element in the set ${\bf NCP}_{\id}(L^\infty(\nu)\otimes\mc L(\mc K);\hil)\equiv\mc I_\nu(\mc K;\hil)$ is called an \emph{instrument} \cite{QTOS76}.
An instrument describes the conditional state changes
experienced by a quantum system in a measurement with the value space $(\Om,\Sigma,\nu)$. Instruments $G\in\mc I_\nu(\mc K;\hil)$ are usually treated as maps $\Gamma:\Sigma\times
\mc L(\mc K)\to\mc L(\hil)$ such that $\Gamma(X,B)=G(\chi_X\otimes B)$ for all $X\in\Sigma$ and $B\in\mc L(\mc K)$. As in the case of observables, instruments can be defined in a broader
context as maps $\Gamma:\Sigma\times\mc L(\mc K)\to\mc L(\hil)$, where $\Gamma(\cdot,B)$ is an operator measure for all $B\in\mc L(\mc K)$, $\Gamma(\cdot,\id)\in\mc O_\Sigma(\hil)$,
$\Gamma(X,\cdot)$ is a normal linear CP map and $\Gamma(\Om,\cdot)\in\mc C(\mc K;\hil)$. 
We denote the set of such maps by $\mc I_\Sigma(\mc K;\hil)$. The dilation space of a minimal
Stinespring dilation of an instrument $\Gamma$ is separable if and only if the dilation space of a minimal Na\u{\i}mark dilation of the associated POVM $\Gamma(\cdot,\id)$ is separable, hence the
separability of the dilation space depends on the properties of the value space $(\Om,\Sigma)$.

\end{itemize}

\section{Extremal points}\label{sec:ext}

The set ${\bf CP}_P(\mc{A};\hil)$ is convex; if $\Phi_1,\Phi_2\in {\bf CP}_P(\mc{A};\hil)$, then $t\Phi_1 + (1-t) \Phi_2\in{\bf CP}_P(\mc{A};\hil)$ for all $0\leq t \leq 1$.
An element $\Phi\in{\bf CP}_P(\mc{A};\hil)$ is called  \emph{extremal} (or \emph{extreme} or \emph{pure}) in ${\bf CP}_P(\mc{A};\hil)$ if it cannot be written as a mixture $\Phi=t\Phi_1+(1-t)\Phi_2$ for two different elements
$\Phi_1,\Phi_2 \in{\bf CP}_P(\mc{A};\hil)$ and $0<t<1$. 
It is obvious that when deciding whether $\Phi$ is extremal or not, it suffices to study mixtures with $t=1/2$. 
The following result characterizes the
extremal elements in ${\bf CP}_P(\mc{A};\hil)$ \cite{Arveson69}. We provide a proof for the reader's convenience.

\begin{theorem}\label{ext}
Suppose that $\Phi\in{\bf CP}_P(\mc{A};\hil)$ and let $(\mc M,\pi,J)$ be its minimal Stinespring dilation. The map $\Phi$ is extremal in ${\bf CP}_P(\mc{A};\hil)$ if and only if for any $E\in\mc L(\mc M)$ the
conditions $[E,\pi(a)]=0$ for all $a\in\mc A$ and $J^*EJ=0$ imply $E=0$.
\end{theorem}

\begin{proof}
Suppose that $\Phi\in{\bf CP}_P(\mc{A};\hil)$ is not extremal, i.e., there are $\Phi_\pm\in{\bf CP}_P(\mc{A};\hil)$ such that $\Phi_+\neq\Phi_-$ and $\Phi=\frac12\Phi_++\frac12\Phi_-$. Hence
$\Phi_\pm\leq2\Phi$ and according to Proposition \ref{RNd} there are operators $F_\pm\in\mc L(\mc M)$ that commute with $\pi$ and 
\begin{equation}
\Phi_\pm(a)=2J^*\pi(a)F_\pm J=J^*\pi(a)E_\pm J \qquad \forall a\in\mc{A} \, , 
\end{equation}
where $E_\pm=2F_\pm$. 
Define $E=E_+-E_-$. 
The operator $E$ is nonzero since $\Phi_+\neq\Phi_-$. 
Moreover, $E$ commutes with $\pi$ and
\begin{equation}
J^*EJ=J^*E_+J-J^*E_-J=\Phi_+(1)-\Phi_-(1)=P-P=0 \, .
\end{equation}

Suppose then that $E\neq0$ satisfies $[E,\pi(a)]=0$ for all $a\in\mc{A}$ and $J^*EJ=0$. 
We may assume that $E$ is selfadjoint; if it is not selfadjoint we may redefine $E'=i(E^*-E)$. 
We may also assume that $\|E\|\leq1$; otherwise we can redefine $E''=\|E\|^{-1}E$. 
Thus we are free to assume that $-\id_{\mc{M}}\leq E\leq \id_{\mc{M}}$ and we can define the positive operators
$E_\pm=\id_{\mc{M}}\pm E$ that commute with the representation $\pi$ and the CP maps $\Phi_\pm\in{\bf CP}(\mc{A};\hil)$ through
\begin{equation}
\Phi_\pm(a)=J^*\pi(a)(\id_{\mc{M}}\pm E)J=J^*\sqrt{\id_{\mc{M}}\pm E}\pi(a)\sqrt{\id_{\mc{M}}\pm E}J \qquad \forall a\in\mc{A} \, .
\end{equation}
Using the properties of $E$ we obtain
\begin{equation}
\Phi_\pm(1)=J^*(\id_{\mc{M}}\pm E)J=J^*J\pm J^*EJ=J^*J=\Phi(1)=P \, .
\end{equation}
Thus $\Phi_\pm\in{\bf CP}_P(\mc{A};\hil)$. 
Since $E\neq0$ and the dilation $(\mc M,\pi,J)$ is minimal, it follows that for some $\f,\,\psi\in\hil$ and $a,\,b\in\mc{A}$ we have
\begin{equation}
0\neq\sis{\pi(a)J\f}{E\pi(b)J\psi}=\frac12\Sis{\f}{\big(\Phi_+(a^*b)-\Phi_-(a^*b)\big)\psi} \, .
\end{equation}
Thus $\Phi_+\neq\Phi_-$ and $\Phi$ has a nontrivial convex decomposition $\Phi=\frac12\Phi_+
+\frac12\Phi_-$. 
Therefore, $\Phi$ is not extremal in ${\bf CP}_P(\mc{A};\hil)$.
\end{proof}

Theorem \ref{ext} can also be applied to the convex set ${\bf NCP}_P(\mc{A};\hil)$.
Namely, {\it a map $\Phi\in{\bf NCP}_P(\mc{A};\hil)$ associated with a minimal Stinespring dilation $(\mc{M},\pi,J)$ is
extremal in ${\bf NCP}_P(\mc{A};\hil)$ if and only if for an operator $E\in\mc L(\mc M)$ the conditions $[E,\pi(a)]=0$ for all $a\in\mc{A}$ and $J^*EJ=0$ yield $E=0$.} This is due to the fact that when
$\Phi\in{\bf NCP}_P(\mc{A};\hil)$, then any map $\Psi\in{\bf CP}_P(\mc{A};\hil)$ appearing in a convex decomposition of $\Phi$ is necessarily normal.

Theorem \ref{ext} gives extremality conditions for the usual physically relevant CP maps. 
We recall the following characterizations.

\begin{itemize}

\item \emph{States}: Since the set of states can be associated with ${\bf NCP}_1(\mc L(\hil);\mb C)$ the extremal points are
characterized by Theorem \ref{ext}; {\it a state $\varrho$ is extremal in $\mc S(\hil)$ if and only if $\varrho$ is a one-dimensional projection.}
This is, of course, a classical and well-known result. 
An extremal state is usually called {\it pure}.

\item \emph{Observables}: Let $T$ be a trace class operator satisfying $\sis{\f}{T\f}>0$ for all $\f\in\hil\setminus\{0\}$. Suppose that $\ms M\in\mc O_\Sigma(\hil)$ and define a finite measure $\nu_T^{\ms M}=\tr{T\ms M(\cdot)}$. One sees that the complex measures $X\mapsto
\sis{\f}{\ms M(X)\psi}$, $\f,\,\psi\in\hil$, are absolutely continuous with respect to $\nu_T^{\ms M}$, or briefly $\ms M\ll \nu_T^{\ms M}$. Hence, $\ms M\in\mc O_{\nu_T^{\ms M}}(\hil)$.
If $\ms M\in\mc O_\Sigma(\hil)$ and $\ms M\ll \nu$, where $\nu:\Sigma\to[0,\infty]$ is a $\sigma$-finite measure, it follows that whenever $\ms N\in\mc O_\Sigma(\hil)$ appears in a convex
decomposition of $\ms M$, then $\ms N\ll \nu$. Thus, we may apply Theorem \ref{ext} to the larger set $\mc O_\Sigma(\hil)$: {\it An observable
$\ms M\in\mc O_\Sigma(\hil)$ with a minimal Na\u{\i}mark dilation $(\mc M,\ms P,J)$, where $\mc M$ is a Hilbert space, $\ms P\in\mc O_\Sigma(\mc M)$ a projection-valued measure and $J:\hil\to
\mc M$ is an isometry so that $\ms M(X)=J^*\ms P(X)J$, is extremal in $\mc O_\Sigma(\hil)$ if and only if for an operator $E\in\mc L(\mc M)$ the conditions $[E,\ms P(X)]=0$ for all $X\in\Sigma$ and
$J^*EJ=0$ yield $E=0$.}
This result has been also proved in \cite{Pellonpaa11}.

\item \emph{Channels}: Let us consider a channel $\mc E\in\mc C(\mc K;\hil)$. 
Assume that $(\mc M,\pi,J)$ is a minimal Stinespring dilation of $\mc E$. The *-representation $\pi$ of $\mc L(\mc K)$ is normal and hence it is unitarily
equivalent with a direct sum of the identity map. This means that we may choose $\mc M=\mc K\otimes\mc M_0$, where $\mc M_0$ is a separable Hilbert space (often called as an {\it ancilla})
and $\pi(B)=B\otimes \id_{\mc M_0}$ for all $B\in\mc L(\mc K)$. We often identify the triple $(\mc M,\pi,J)$ with the pair $(\mc M_0,J)$. The set of those operators in $\mc L(\mc M)$ that
commute with the representation $\pi$ are of the form $\id_{\mc K}\otimes D$ with $D\in\mc L(\mc M_0)$. We may characterize the set of extremal channels in the following way: {\it A channel
$\mc E\in\mc C(\mc K;\hil)$ with the minimal dilation $(\mc M_0,J)$ is extremal in $\mc C(\mc K;\hil)$ if and only if for any $D\in\mc L(\mc M_0)$ the condition $J^*(\id_{\mc K}\otimes D)J=0$
implies $D=0$.}

We may characterize the extremality of a channel $\mc E\in\mc C(\mc K;\hil)$ also in another way, first proved in \cite{Choi75} in the case of finite Hilbert spaces.
Suppose that $(\mc M_0,J)$ is a minimal dilation of $\mc E$ as above.
We fix an orthonormal basis $\{\xi_j\}_{j=1}^{\dim{(\mc M_0)}}$ and define the operators $R_j:\mc K\otimes\mc M_0\to\mc K$, such that $R_j(\f\otimes\psi)=\sis{\xi_j}{\psi}\f$ for all $\f\in\mc K$ and
$\psi\in\mc M_0$. 
For all $B\in\mc L(\mc K)$, we may write
\begin{eqnarray*}
\mc E(B)&=&J^*(B\otimes \id_{\mc M_0})J=\sum_{j=1}^{\dim{(\mc M_0)}}J^*(B\otimes\ktb{\xi_j}{\xi_j})J\\
&=&\sum_{j=1}^{\dim{(\mc M_0)}}J^*R_j^*BR_jJ=\sum_{j=1}^{\dim{(\mc M_0)}}K_j^*BK_j,
\end{eqnarray*}
where we have defined $K_j=R_jJ$. The operators $K_j$ are said to set up a {\it Kraus decomposition} for $\mc E$, and when the operators are obtained as above from a minimal dilation, the
decomposition or the set $\{K_j\}_j$ is said to be {\it minimal}. Suppose that $D\in\mc L(\mc M_0)$ and define $d_{j,k}=\sis{\xi_j}{D\xi_k}$. Inserting the identity operator $\sum_j\ktb{\xi_j}{\xi_j}$
on both sides of $D$, one finds that 
\begin{equation}
J^*(\id_{\mc K}\otimes D)J=\sum_{j,k}d_{j,k}K_j^*K_k \, .
\end{equation}
A set $\{B_{j,k}\}_{j,k}$ of bounded operators on $\mc K$ is said to be {\it strongly independent}
if $\sum_{j,k}\alpha_{j,k}B_{j,k}=0$ implies $\alpha_{j,k}=0$ for all $j,\,k$ whenever $(\alpha_{j,k})_{j,k=1}^{\dim{(\mc K)}}$ is a matrix representation of a bounded operator on $\mc K$. We may
rephrase the extremality condition of channels: {\it A channel $\mc E\in\mc C(\mc K;\hil)$ with a minimal set $\{K_j\}_j$ of Kraus operators is extremal in $\mc C(\mc K;\hil)$ if and only if the
set $\{K_j^*K_k\}_{j,k}$ is strongly independent.} 
This result can be found in \cite{Tsui96}.

\item \emph{Instruments}:
A complete characterization of extremal instruments is presented in \cite{Part1}.
It can be stated in a similar form as for channels by using minimal pointwise Kraus decompositions of instruments.
Actually, all the above characterizations are special cases of this general characterization due to the fact that states, observables, and channels can be viewed as instruments \cite{Part1}.

\end{itemize}

\section{Marginal maps and joint maps}\label{sec:joint}

In this section $\mc A$ and $\mc B$ are von Neumann algebras so that we may define their von Neumann tensor product $\mc A\otimes\mc B$. 
We also assume that $P\in\mc L(\hil)$ is fixed. 
We study maps belonging to ${\bf CP}_P(\mc{A};\hil)$, ${\bf CP}_P(\mc{B};\hil)$ or ${\bf CP}_P(\mc{A}\otimes\mc{B};\hil)$. 
We give our main result (Theorem \ref{extmarg}) without assuming normality. 
According to the discussion in the end of the previous section, all results apply also in the case of normal maps.

\begin{definition}
{\rm
Suppose that $\Psi\in{\bf CP}_P(\mc A\otimes\mc B;\hil)$. 
The maps 
\begin{equation}
\Psi_{(1)}:\mc A\to\mc L(\hil) \, , \quad \Psi_{(1)}(a)=\Psi(a\otimes \id_{\mc B}) \quad \textrm{for all $a\in\mc A$}
\end{equation}
and
\begin{equation}
\Psi_{(2)}:\mc B\to\mc L(\hil) \, , \quad \Psi_{(2)}(b)=\Psi(\id_{\mc A}\otimes b) \quad \textrm{for all $b\in\mc B$} 
\end{equation}
belong to ${\bf CP}_P(\mc A;\hil)$ and ${\bf CP}_P(\mc B;\hil)$, respectively. 
We call the maps $\Psi_{(1)}$ and $\Psi_{(2)}$ as the {\it marginals of $\Psi$}. Suppose that $\Phi_1\in{\bf CP}_P(\mc A;\hil)$ and $\Phi_2\in{\bf CP}_P(\mc B;\hil)$. If $\Phi_1$ and $\Phi_2$
are marginals of some $\Psi\in{\bf CP}_P(\mc A\otimes\mc B;\hil)$, i.e., $\Phi_1=\Psi_{(1)}$ and $\Phi_2=\Psi_{(2)}$, we say that $\Phi_1$ and $\Phi_2$ are {\it compatible}. Moreover, we say that $\Psi$ is a {\it joint map} for $\Phi_1$ and $\Phi_2$.
}
\end{definition}

To demonstrate the content of this definition, let us recall two common instances of compatibility.
Suppose that $\mc K$ is a Hilbert space, $(\Om,\Sigma)$ is a measurable space, and $\nu:\Sigma\to[0,\infty]$ is a $\sigma$-finite measure.
We set $\mc A=L^\infty(\nu)$, $\mc B=\mc L(\mc K)$ and $P=\id_\hil$.
If an observable $\ms M\in\mc O_\nu(\hil)$ and a channel $\mc E\in\mc C(\mc K;\hil)$ are compatible, we say that their joint
map $\Gamma\in\mc I_\nu(\mc K;\hil)$ is their {\it joint instrument}. 
Clearly, every instrument is a joint instrument of some observable and channel. 
For instance, let $\Om=\nat$ and $\Gamma(X,B)=\sum_{j\in X} K_j^\ast B K_j$ for some set of operators $\{K_j\}_j$.
Then the marginals of $\Gamma$ are $\ms M(X)=\sum_{j\in X} K_j^\ast K_j$ and $\mc E(B)=\sum_{j\in\Om} K_j^\ast B K_j$.

For another example, we set $\mc A=L^\infty(\nu)$ and $\mc B=L^\infty(\nu')$, where 
 $(\Om',\Sigma')$ is a measurable space and $\nu':\Sigma'\to[0,\infty]$ is a $\sigma$-finite measure.
If two observables $\ms M_1\in\mc O_\nu(\hil)$ and $\ms M_2\in
\mc O_{\nu'}(\hil)$ are compatible, we say that their joint map $\ms M\in\mc O_{\nu\times\nu'}(\hil)$ is their {\it joint observable}. 
Compatible observables are usually called {\it jointly measurable}.
According to the remarks made in the previous section, all the results concerning observables and instruments in this section also apply to the entire sets $\mc O_\Sigma(\hi)$,
$\mc O_{\Sigma'}(\hi)$ and $\mc I_\Sigma(\mc K;\hi)$. It should be noted that a joint observable $\ms M$ of $\ms M_1\in\mc O_\Sigma(\hil)$ and $\ms M_2\in\mc O_{\Sigma'}(\hil)$ is in $\mc O_{\Sigma
\otimes\Sigma'}(\hil)$, where $\Sigma\otimes\Sigma'$ is the product $\sigma$-algebra of $\Sigma$ and $\Sigma'$.
The fact that $\ms M$ is a joint observable means that
$\ms M(X\times \Omega)=\ms M_1(X)$ and $\ms M(\Omega' \times Y)=\ms M_2(Y)$ for all  $X\in\Sigma$, $Y\in\Sigma'$.

The main result of our investigation is the following.

\begin{theorem}\label{extmarg}
Suppose that $\Phi_1\in{\bf CP}_P(\mc A;\hil)$ and $\Phi_2\in{\bf CP}_P(\mc B;\hil)$ are compatible.
\begin{itemize}
\item[{\rm (a)}] If $\Phi_1$ is extremal in ${\bf CP}_P(\mc A;\hil)$ or $\Phi_2$ is extremal in ${\bf CP}_P(\mc B;\hil)$, then they have a unique joint map.
\item[{\rm (b)}] If $\Phi_1$ is extremal in ${\bf CP}_P(\mc A;\hil)$ and $\Phi_2$ is extremal in ${\bf CP}_P(\mc B;\hil)$, then their unique joint map is extremal in ${\bf CP}_P(\mc A\otimes\mc B;\hil)$.
\item[{\rm (c)}] If $\Phi_1$ or $\Phi_2$ is a *-representation, then $\Phi_1$ and $\Phi_2$ commute and the unique joint map $\Psi\in{\bf CP}(\mc A\otimes\mc B;\hil,P)$ is of the form
\begin{equation}\label{tahtiesit}
\Psi(a\otimes b)=\Phi_1(a)\Phi_2(b),\qquad a\in\mc A,\quad b\in\mc B.
\end{equation}
\end{itemize}
\end{theorem}

Before we present a proof of Theorem \ref{extmarg}, we need some auxiliary results.

\begin{lemma}\label{submin}
Suppose that $\Psi\in{\bf CP}_P(\mc A\otimes\mc B;\hil)$ and $(\mc M,\pi,J)$ is a minimal Stinespring dilation of the first marginal $\Psi_{(1)}$. There is a unique map
$E\in{\bf CP}_\id(\mc B;\mc M)$ such that $[\pi(a),E(b)]=0$ for all $a\in\mc A$, $b\in\mc B$ and
\begin{equation}\label{tulo}
\Psi(a\otimes b)=J^*\pi(a)E(b)J,\qquad a\in\mc A,\quad b\in\mc B \, .
\end{equation}
\end{lemma}

\begin{proof}
Fix $b\in\mc B$ and denote the map $\mc A\ni a\mapsto\Psi(a\otimes b)$ by $\Psi_b$. Clearly, if $b\geq0$, then $\Psi_b\in{\bf CP}(\mc A;\hil)$. Furthermore
\begin{equation}
\|b\|\Psi_{(1)}(a)-\Psi_b(a)=\Psi\big(a\otimes(\|b\|\id_{\mc B}-b)\big),
\end{equation}
and hence $\Psi_b\leq\|b\|\Psi_{(1)}$ for any positive $b\in\mc B$. 
Thus, according to Proposition \ref{RNd}, for any positive $b\in\mc B$, we have an operator $E(b)\in\mc L(\mc M)$ that commutes
with $\pi$ such that $\Psi_b(a)=J^*\pi(a)E(b)J$. We may uniquely extend this construction into a linear map $E:\mc B\to\mc L(\mc M)$ such that $[\pi(a),E(b)]=0$ and \eqref{tulo} holds.
The unitality of $E$ is obvious since 
\begin{equation}
J^*\pi(a)J=\Psi_{(1)}(a)=J^*\pi(a)E(\id_{\mc B})J
\end{equation}
and the choice of $E(\id_{\mc B})$ is unique by the minimality of the dilation.

Pick natural numbers $n$, $n_j$ and vectors $\f_{jk}\in\hil$ and $a_{jk}\in\mc A$, $k=1,\ldots,n_j$ $j=1,\ldots,n$ and $b_1,\ldots,b_n\in\mc B$ and define 
\begin{equation}
\eta_j=\sum_{k=1}^{n_j}\pi(a_{jk})J\f_{jk}\in\mc M.
\end{equation}
We may evaluate
\begin{eqnarray*}
\sum_{j,k=1}^n\sis{\eta_j}{E(b_j^*b_k)\eta_k}&=&\sum_{j,k=1}^n\sum_{l=1}^{n_j}\sum_{m=1}^{n_k}\sis{J\f_{jl}}{\pi(a_{jl}^*)E(b_j^*b_k)\pi(a_{km})J\f_{km}}\\
&=&\sum_{j,k=1}^n\sum_{l=1}^{n_j}\sum_{m=1}^{n_k}\sis{J\f_{jl}}{\pi(a_{jl}^*a_{km})E(b_j^*b_k)J\f_{km}}\\
&=&\sum_{j,k=1}^n\sum_{l=1}^{n_j}\sum_{m=1}^{n_k}\sis{\f_{jl}}{\Psi\big((a_{jl}\otimes b_j)^*(a_{km}\otimes b_k)\big)\f_{km}}\geq0
\end{eqnarray*}
using the complete positivity of $\Psi$. Again, the minimality of $(\mc M,\pi,J)$ implies the complete positivity of $E$.
\end{proof}

\begin{definition}\label{subminma}
{\rm
Suppose that $\Psi\in{\bf CP}_P(\mc A\otimes\mc B;\hil)$ is associated with a quadruple $(\mc M,\pi,E,J)$ where $(\mc M,\pi,J)$ is a minimal dilation of the first marginal $\Psi_{(1)}$ and
$E\in{\bf CP}_I(\mc B;\mc M)$ as in Lemma \ref{submin}.
We call the quadruple $(\mc M,\pi,E,J)$ as an {\it $\mc A$-subminimal dilation of $\Psi$}. We define the $\mc B$-subminimal dilation in
the same way.
}
\end{definition}

Let $\Psi\in{\bf CP}_P(\mc A\otimes\mc B;\hil)$ and $(\mc M,\pi,E,J)$ be an $\mc A$-subminimal dilation of $\Psi$. 
By extending the map $a\otimes b\mapsto\pi(a)E(b)$ into a (unique) CP map
$F\in{\bf CP}_\id(\mc A\otimes\mc B;\mc M)$, we have $\Psi(c)=J^*F(c)J$ for all $c\in\mc A\otimes\mc B$. Typically, the map $F$ is not a *-representation; hence the name `subminimal dilation'.

\begin{proof}[Proof of Theorem \ref{extmarg}.]
(a) Suppose that $\Psi$ and $\Psi'$ are joint maps for $\Phi_1$ and $\Phi_2$ and assume that $\Phi_1$ is extremal in ${\bf CP}(\mc A;\hil,P)$ and
$(\mc M,\pi,J)$ is a minimal dilation of $\Phi_1$. Assume that $(\mc M,\pi,E,J)$ and $(\mc M,\pi,E',J)$ are $\mc A$-subminimal dilations for $\Psi$ and respectively for $\Psi'$.
We may write
\begin{equation}
0=\Phi_2(b)-\Phi_2(b)=\Psi(\id_{\mc A}\otimes b)-\Psi'(\id_{\mc A}\otimes b)=J^*\big(E(b)-E'(b)\big)J
\end{equation}
for all $b\in\mc B$. Since the operator $E(b)-E'(b)$ commutes with the representation $\pi$ for all $b\in\mc B$, the extremality of $\Phi_1$ yields $E=E'$. This means that $\Psi(a\otimes b)=
\Psi'(a\otimes b)$ for all $a\in\mc A$ and $b\in\mc B$ and thus $\Psi=\Psi'$.

(b) Suppose that both $\Phi_1$ and $\Phi_2$ are extremal. Item (a) yields the uniqueness of the joint map $\Psi\in{\bf CP}(\mc A\otimes\mc B;\hil,P)$. Let us
make a counter assumption: there are $\Psi^\pm\in{\bf CP}(\mc A\otimes\mc B;\hil,P)$ such that $\Psi^+\neq\Psi^-$ and $\Psi=\frac12\Psi^++\frac12\Psi^-$. Hence also $\Phi_r=
\frac12\Psi^+_{(r)}+\frac12\Psi^-_{(r)}$, $r=1,\,2$. Suppose that $(\mc M,\pi,J)$ is a minimal dilation of $\Phi_1$. Since, especially, $\Psi^\pm_{(1)}\leq2\Phi_1$, combining
the result of Proposition \ref{RNd} and the beginning of the proof of Lemma \ref{submin}, we obtain CP maps $E^\pm:\mc B\to\mc L(\mc M)$ such that $\Psi^\pm(a\otimes b)=J^*\pi(a)
E^\pm(b)J$ for all $a\in\mc A$ and $b\in\mc B$. Since $\Psi^\pm$ are characterized by the quadruples $(\mc M,\pi,E^\pm,J)$, it follows that $E^+\neq E^-$. Suppose that
$b\in\mc B$ is such that $E^+(b)\neq E^-(b)$. Since $\Phi_2$ is extremal, it follows that $\Psi^+_{(2)}=\Psi^-_{(2)}$ and thus
$$
0=\Psi^+_{(2)}(b)-\Psi^-_{(2)}(b)=\Psi^+(\id_{\mc A}\otimes b)-\Psi^-(\id_{\mc A}\otimes b)=J^*\big(E^+(b)-E^-(b)\big)J \, .
$$
The extremality of $\Phi_1$ thus yields that $E^+(b)=E^-(b)$, contradicting our assumption on $b$.

(c) Suppose that $\Phi_1$ is associated with a minimal dilation $(\mc M,\pi,J)$ yielding the $\mc A$-subminimal dilation $(\mc M,\pi,E,J)$ so that (\ref{tulo}) holds.
Let us assume that $\Phi_1$ is a *-homomorphism. Then $J$ is unitary and we may simply choose $\mc M=\hil$, $\pi=\Phi_1$ and $J=I$. It is clear that $\Phi_1$ is extremal and hence the joint
map $\Psi$ is unique. From (\ref{tulo}) it immediately follows that $E=\Phi_2$ and thus $\Phi_2$ commutes with $\Phi_1$ and equation (\ref{tahtiesit}) holds.
\end{proof}

The statements (a)--(f) sketched in the introduction are now consequences of the previous theorem. 
Note that, when viewed as CP maps, any two states $\varrho_1\in\mc S(\hil_1)$ and $\varrho_2\in\mc S(\hil_2)$ have a joint state $\varrho_1\otimes\varrho_2 \in \mc{S}(\hil_1\otimes\hil_2)$.

\begin{cor}
Suppose that $\mc K$, $\hil$, $\hil_1$ and $\hil_2$ are Hilbert spaces, $(\Om,\Sigma)$ and $(\Om',\Sigma')$ are measurable spaces and $\nu:\Sigma\to[0,\infty]$ and $\nu':\Sigma'\to[0,
\infty]$ are $\sigma$-finite measures.
\begin{itemize}
\item[{\rm (a)}] Suppose that $\varrho\in\mc S(\hil_1\otimes\hil_2)$. 
If $\mr{tr}_{\hil_1}[\varrho]=\varrho_2$ or $\mr{tr}_{\hil_2}[\varrho]=\varrho_1$ is pure, then $\varrho=\varrho_1\otimes\varrho_2$.
\item[{\rm (b)\&\rm (e)}] Suppose that $\ms M_1\in\mc O_\nu(\hil)$ and $\ms M_2\in\mc O_{\nu'}(\hil)$ are jointly measurable. If $\ms M_1$ is extremal in $\mc O_\nu(\hil)$ or $\ms M_2$ is extremal in
$\mc O_{\nu'}(\hil)$, then their joint observable $\ms M\in\mc O_{\nu\times\nu'}(\hil)$ is unique. If both $\ms M_1$ and $\ms M_2$ are extremal, then the joint observable is extremal in
$\mc O_{\nu\times\nu'}(\hil)$. If $\ms M_1$ or $\ms M_2$ is sharp, then $\ms M_1$ and $\ms M_2$ commute and their unique joint observable $\ms M\in\mc O_{\nu\times\nu'}(\hil)$ is determined by the condition
$$
\ms M(X\times Y)=\ms M_1(X)\ms M_2(Y),\qquad X\in\Sigma,\quad Y\in\Sigma'.
$$
\item[{\rm (c)\&\rm (f)}] Suppose that $\ms M\in\mc O_\nu(\hil)$ and $\mc E\in\mc C(\mc K;\hil)$ are compatible. If $\ms M$ is extremal in $\mc O_\nu(\hil)$ or $\mc E$ is extremal in $\mc C(\mc K;\hil)$,
then their joint instrument $\Gamma\in\mc I_\nu(\mc K;\hil)$ is unique. If both $M$ and $\mc E$ are extremal, then this $\Gamma$ is extremal in $\mc I_\nu(\mc K;\hil)$. If $\ms M$ is sharp or $\mc E$
is a *-representation, then $[\ms M(X),\mc E(B)]=0$ for all $X\in\Sigma$ and $B\in\mc L(\mc K)$ and the unique joint instrument $\Gamma$ for $\ms M$ and $\mc E$ is given by
$$
\Gamma(X,B)=\ms M(X)\mc E(B)=\sqrt{\ms M(X)}\mc E(B)\sqrt{\ms M(X)},\qquad X\in\Sigma,\quad B\in\mc L(\mc K)
$$
\item[{\rm (d)}] Suppose that $\mc F:\mc T(\hil)\to\mc T(\hil)\otimes\mc T(\hil)$ is a Schr\"odinger channel such that the partial trace $\mr{tr}_1[\mc F(T)]=T$ for all $T\in\mc T(\hil)$. There
is a state $\sigma\in\mc S(\hil)$ such that $\mr{tr}_2[\mc F(T)]=\tr T\,\sigma$ for all $T\in\mc T(\hil)$.

\end{itemize}
\end{cor}

\begin{proof}
Items (b)\&(e) and (c)\&(f) are direct consequences of Theorem \ref{extmarg}. 
Let us first concentrate on item (a). Given any Hilbert space $\hil$, the set $\mc S(\hil)$ of state operators on $\hil$ is in
bijective affine correspondence with ${\bf NCP}(\mc L(\hil);\mb C,1)$ when one defines the map $\mc S(\hil)\ni\varrho\mapsto\Phi^\varrho$, $\Phi^\varrho(C)=\tr{\varrho C}$, $C\in\mc L(\hil)$. Suppose that
$\varrho\in\mc S(\hil_1\otimes\hil_2)\simeq{\bf NCP}(\mc L(\hil_1)\otimes\mc L(\hil_2);\mb C,1)$. The marginals of the corresponding positive map $\Phi^\varrho$ are $\Phi^\varrho_r$, $\Phi^\varrho_1(A)=
\Phi^\varrho(A)=\tr{\varrho(A\otimes I_{\hil_2})}=\tr{\varrho_1A}$ and similarly $\Phi^\varrho_2(B)=\tr{\varrho_2B}$ for all $A\in\mc L(\hil_1)$ and $B\in\mc L(\hil_2)$, where $\varrho_1=\mr{tr}_{\hil_2}[\varrho]$
and $\varrho_2=\mr{tr}_{\hil_1}[\varrho]$. It now follows from Theorem \ref{extmarg} that if either one of the marginals, i.e.\ partial traces, of a state $\varrho\in\mc S(\hil_1\otimes\hil_2)$ is pure, then
there is no other state $\varrho'\in\mc S(\hil_1\otimes\hil_2)$ that has the same marginals. Since especially the state $\varrho_1\otimes\varrho_2$ has the marginals $\varrho_1$ and $\varrho_2$, the claim
follows.

Let us prove item (d). Denote the transpose (Heisenberg channel) of $\mc F$ by $\mc E$, i.e.\ $\mc E\in\mc C(\hil\otimes\hil;\hil)$. We may write
$$
\tr{\mc E_{(1)}(A)T}=\tr{\mc E(A\otimes I)T}=\tr{(A\otimes I)\mc F(T)}=\tr{\mr{tr}_1[\mc F(T)]A}=\tr{TA}
$$
for all $A\in\mc L(\hil)$ and $T\in\mc T(\hil)$. Hence $\mc E_{(1)}$ is the identity map and especially a *-homomorphism. Thus, according to item (c) of Theorem \ref{extmarg}, the second
marginal $\mc E_{(2)}$ must take values in the center of $\mc L(\hil)$, i.e.\ there is a state operator $\sigma\in\mc S(\hil)$ such that $\mc E_{(2)}(B)=\tr{\sigma B}\id$ for all $B\in\mc L(\hil)$. 
One easily sees that this means $\mr{tr}_2[\mc F(T)]=\tr{T}\sigma$ for all $T\in\mc T(\hil)$.
\end{proof}

\section{Causal channels on bipartite systems}\label{sec:comp}

In this section we study a special consequence of the results obtained in the previous sections to the structure of channels on bipartite systems.
Before we switch to the channel terminology, we present the essential result in a general form.
  
Suppose that $\mc A$ and $\mc B$ are von Neumann algebras, $\hil_1$ and $\hil_2$ are Hilbert
spaces, and $P_1\in \lin{\hil_1}$ and $P_2\in \lin{\hil_2}$ are positive operators.
Let $\Phi_1\in{\bf CP}_{P_1}(\mc A;\hil_1)$ and $\Phi_2\in{\bf CP}_{P_2}(\mc B;\hil_2)$. 
There is a single CP map
$\Psi\in{\bf CP}_{P_1\otimes P_2}(\mc A\otimes\mc B;\hil_1\otimes\hil_2)$ such that $\Psi(a\otimes b)=\Phi_1(a)\otimes\Phi_2(b)$ for all $a\in\mc A$ and $b\in\mc B$, and we denote $\Psi\equiv\Phi_1\otimes\Phi_2$. 
We denote the marginals of $\Phi_1\otimes\Phi_2$ by $\tilde\Phi_1$ and $\tilde\Phi_2$, i.e.\ $\tilde\Phi_1(a)=\Phi_1(a)\otimes P_2$ and $\tilde\Phi_2(b)=P_1\otimes\Phi_2(b)$. 

\begin{proposition}\label{laajext}
Retain the notations defined above. If $\Phi_1$ is extremal in ${\bf CP}_{P_1}(\mc A;\hil_1)$, then $\tilde\Phi_1$ is extremal in ${\bf CP}_{P_1\otimes P_2}(\mc A\otimes\mc B;\hil_1\otimes
\hil_2)$.
\end{proposition}

\begin{proof}
We prove the claim by showing that if $\tilde\Phi_1$ is {\it not} extremal in ${\bf CP}_{P_1\otimes P_2}(\mc A\otimes\mc B;\hil_1\otimes\hil_2)$, then $\Phi_1$ is not extremal in ${\bf CP}_{P_1}(\mc A;\hil_1)$.
Assume that $(\mc M,\pi,J)$ is a minimal dilation of $\Phi_1$.
Further assume that $\mc M_2$ is a Hilbert space and $J_2:\hil_2\to\mc M_2$ is a linear map such that the image space $J_2(\hil_2)$ is dense in
$\mc M_2$ and $P_2=J_2^*J_2$.  
(A possible choice is, for instance, $J_2=\sqrt{P_2}$ and $\mc M_2$ is the closure of the range of $\sqrt{P_2}$.)
We denote by $\tilde\pi$ the map defined by $\tilde{\pi}(a)=\pi(a)\otimes\id_{\mc M_2}$, $a\in\mc{A}$.
Then $(\mc M\otimes\mc M_2,\tilde\pi,J\otimes J_2)$ is a minimal dilation for $\tilde\Phi_1$. 
For each $\psi\in\mc M_2$, we define the linear operator
$R_\psi:\mc M\otimes\mc M_2\to\mc M$ through 
\begin{equation}
R_\psi(\zeta\otimes\xi)=\sis{\psi}{\xi}\zeta
\end{equation}
 for all $\zeta\in\mc M$ and $\xi\in\mc M_2$. 
 Pick any $\zeta\in\mc M$, $\psi\in\mc M_2$ and
$a\in\mc A$. We have
\begin{equation}
R_\psi^*\pi(a)\zeta=\big(\pi(a)\zeta\big)\otimes\psi=\tilde\pi(a)(\zeta\otimes\psi)=\tilde\pi(a)R_\psi^*\zeta.
\end{equation}
This yields 
\begin{equation}
R_\psi^*\pi(a)=\tilde\pi(a)R_\psi^* \, , \qquad
\pi(a)R_\psi=R_\psi\tilde\pi(a)
\end{equation}
for all $\psi\in\mc M_2$ and $a\in\mc A$.

Suppose that $D\in\mc L(\mc M\otimes\mc M_2)$, $D\neq0$, is such that $\tilde\pi(a)D=D\tilde\pi(a)$ for all $a\in\mc A$ and $(J\otimes J_2)^*D(J\otimes J_2)=0$. Since $D$ is nonzero,
there are $\chi,\,\psi\in\mc M_2$ such that $R_\chi DR_\psi^*\neq0$; this is obvious because $R_\psi^*R_\psi=\id_{\mc M}\otimes \ktb{\psi}{\psi}$. Using the properties of $R_\psi$ and $R_\chi$ established above, one obtains
\begin{equation}
R_\chi DR_\psi^*\pi(a)=R_\chi D\tilde\pi(a)R_\psi^*=R_\chi\tilde\pi(a)DR_\psi^*=\pi(a)R_\chi DR_\psi^* \, , 
\end{equation}
i.e.,\ the operator $R_\chi DR_\psi^*$ commutes with $\pi$. We also obtain
\begin{equation}
\sis{\f}{J^*R_\chi DR_\psi^*J\f}=\sis{J\f\otimes\chi}{D(J\f\otimes\psi)}=0 \, , 
\end{equation}
where the last equality follows from the last defining property of $D$ and the fact that $\chi$ and $\psi$ can be approximated with vectors from the image space $J_2(\mc K_2)$. Hence
$D':=R_\chi DR_\psi^*\in\mc L(\mc M)$ is a nonzero operator that commutes with $\pi$ and $J^*D'J=0$. According to Theorem \ref{ext}, $\Phi_1$ is not extremal in ${\bf CP}_{P_1}(\mc A;
\hil_1)$.
\end{proof}

Suppose that $\hil_1,\hil_2$ and $\hik_1,\hik_2$ are separable Hilbert spaces. 
Let us study channels $\mc E\in\mc C(\mc K_1\otimes\mc K_2;\hil_1\otimes\hil_2)$. 
This kind of channel acting on a compound system may have independent action on the first system in the
sense that there exists a channel $\mc E^1\in\mc C(\mc K_1;\hil_1)$ such that the equivalent conditions
\begin{eqnarray}\label{2rton}
\mc E(A\otimes \id)&=&\mc E^1(A)\otimes \id \, , \quad A\in\mc L(\mc K_1) \qquad \mr{or}\nonumber\\
\mr{tr}_{\mc K_2}[\mc E_*(S\otimes T)]&=&\tr{T}\mc E^1_*(S) \, , \quad S\in\mc T(\hil_1),\ T\in\mc T(\hil_2)
\end{eqnarray}
hold. 
On the other hand, there could be a channel $\mc E^2\in\mc C(\mc K_2;\hil_2)$ such that
\begin{eqnarray}\label{1rton}
\mc E(\id\otimes B)&=&\id\otimes\mc E^2(B),\quad B\in\mc L(\mc K_2)\qquad \textrm{or equivalently}\nonumber\\
\mr{tr}_{\mc K_1}[\mc E_*(S\otimes T)]&=&\tr{S}\mc E^2_*(T) \, , \quad S\in\mc T(\hil_1),\ T\in\mc T(\hil_2) \, .
\end{eqnarray}
With the terminology used in \cite{EgScWe02}, we say that a channel satisfying the conditions \eqref{2rton} and \eqref{1rton} is {\it causal}.

A special class of causal channels are local channels.
A channel $\mc E\in\mc C(\mc K_1\otimes\mc K_2;\hil_1\otimes\hil_2)$ is \emph{local} if there exist channels $\mc E^1\in\mc C(\mc K_1;\hil_1)$ and $\mc E^2\in\mc C(\mc K_2;\hil_2)$ such that the equivalent conditions
\begin{eqnarray}\label{localchannel}
\mc E(A\otimes B)&=&\mc E^1(A)\otimes\mc E^2(B) \, , \qquad A\in\mc L(\mc K_1) \, , \   B\in\mc L(\mc K_2)\quad\mr{or}\nonumber\\
\mc E_*(S\otimes T)&=&\mc E^1_*(S)\otimes\mc E^2_*(T) \, , \qquad S\in\mc T(\hil_1) \, , \  T\in\mc T(\hil_2)
\end{eqnarray}
hold.
Given channels $\mc E^1\in\mc C(\mc K_1;\hil_1)$ and $\mc E^2\in\mc C(\mc K_2;\hil_2)$, we can always combine them to the joint local channel $\mc E^1\otimes\mc E^2$, but there might be
other joint channels with the same marginals $a\mapsto\mc E^1(a)\otimes \id_{\hil_2}$ and $b\mapsto \id_{\hil_1}\otimes\mc E^2(b)$. 
However, Proposition \ref{laajext} implies that, in the case of an
extremal marginal, the only joint map is the local one $\mc E^1\otimes\mc E^2$.
In other words, a causal channel with an extremal marginal is local.

\begin{cor}\label{causallocal}
Assume that a channel $\mc E\in\mc C(\mc K_1\otimes\mc K_2;\hil_1\otimes\hil_2)$ is causal, i.e.,\ there are channels $\mc E^1\in\mc C(\mc K_1;\hil_1)$ and $\mc E^2\in\mc C(\mc K_2;\hil_2)$ such that the conditions
(\ref{2rton}) and (\ref{1rton}) hold. If $\mc E^1$ is extremal in $\mc C(\mc K_1;\hil_1)$ or $\mc E^2$ is extremal in $\mc C(\mc K_2;\hil_2)$, then $\mc E$ is local, i.e.,\ $\mc E=\mc E^1\otimes
\mc E^2$.
\end{cor}
 
It is easy to see that convex combinations of local channels are causal. 
If we restrict to the case where $\mc E$ is assumed to be in the convex hull of local channels, then the above result
could be proved without the results on extremal marginals developed in this paper. 
However, the result of Corollary \ref{causallocal} is more extensive since there are causal channels that cannot
be represented as convex combinations of local channels; concrete examples are given in \cite{BeGoNiPr01}, \cite{DaFaPe11}.

\section{Conclusions and Remarks}\label{sec:end}

We have identified extremality of a marginal map as a sufficient criterion for the uniqueness of the joint map. 
This result directly implies several well-known results.
For instance, if either one of two jointly measurable observables is sharp, then they commute and their joint observable is unique
and of the product form. 
Another well-known result states that if a sharp observable $\ms M$ and a channel $\mc E$ are parts of a single instrument, then this instrument $\Gamma$ is unique and $\Gamma(X,B)=\ms M(X)\mc E(B)$.
Our main result states that in the previous results the uniqueness part remains valid also in the case of a looser condition that one of the marginals of the joint map is extremal. 
This is a true expansion of previous results since, for instance, there are extremal observables that are not sharp observables but are physically relevant \cite{HePe09}.

Finally, we emphasize that the extremality of a marginal is not a necessary condition for the uniqueness of the joint map.
Namely, a joint map of two compatible maps can be unique and extremal even if the constituents are not extremal.
To give an example of this kind of occurence, we consider binary observables on a qubit system.
As proved in \cite{Busch86}, the observables $\pm 1 \mapsto \half (I \pm  t \sigma_x)$ and  $\pm 1 \mapsto \half (I \pm  t \sigma_y)$ are compatible exactly when $0\leq t \leq 1/\sqrt{2}$.
It is clear from the proof of \cite[Theorem 4.5.]{Busch86} that in the boundary case $t=1/\sqrt{2}$ the observables have a unique joint observable (the four spheres used in the proof have a single point in their intersection).
A qubit observable is extremal if and only if each effect is rank-1 and the set of effects is linearly independent \cite{DaLoPe05}.
Therefore, it is easy to see that the unique joint observable is extremal while the marginal observables are not.

\section*{Acknowledgements}

Authors wish to thank Takayuki Miyadera and Michal Sedlak for comments on the manuscript and acknowledge financial support from the Academy of Finland (grant no. 138135). E. H. is also
funded by the Finnish Cultural Foundation.

\end{document}